%% file: paper9_final.tex
\documentclass[10pt, conference, letterpaper]{IEEEtran}
\usepackage{amsmath,amssymb,amsfonts}
\usepackage[english]{babel}
\usepackage{blindtext}
\usepackage{booktabs}
\usepackage{cuted}
\usepackage{graphicx,subcaption}
\usepackage{ulem}
\usepackage{caption}
\usepackage{multirow,multicol}
\usepackage{tablefootnote}
\usepackage{bm}
\usepackage{blkarray}
\usepackage[linesnumbered,ruled,vlined]{algorithm2e}
\usepackage{algpseudocode} 
\usepackage{nopageno}

\usepackage{hyperref}
\usepackage{xcolor} 

\newtheorem{theorem}{Theorem}
\newtheorem{claim}{Claim}

\newtheorem{proof}{Proof}

\newenvironment{icompact}{
	\begin{list}{$\bullet$}{
			\parsep 0.5pt plus 0.5pt
			\partopsep 0.5pt plus 0.5pt
			\topsep 0.5pt plus 1pt minus 0.5pt
			\itemsep 0.5pt plus 0.5pt
			\parskip 0pt plus 1pt
			\leftmargin 0.15in}
	}
	{\normalsize\end{list}}

\newcommand{\sys}{{\text{dMAPAR-HMM}}\xspace}

\newlength\heightfiga\newlength\heightcapa
\newlength\heightfigb\newlength\heightcapb
\newlength\heightfigc\newlength\heightcapc
\newlength\heightfig
\newcommand*{\affaddr}[1]{#1} 

\begin{document}
\title{\sys: Reforming Traffic Model for Improving Performance Bound with Stochastic Network Calculus}

\author{%
	Qingqing Yang, Xi Peng, Huiwen Yang, Gong Zhang, and Bo Bai\\
	\affaddr{Theory Lab, Central Research Institute, 2012 Labs, Huawei Technologies Co. Ltd. }\\
	\affaddr{Hong Kong SAR, China}\\
}

\IEEEoverridecommandlockouts
\IEEEpubid{\makebox[\columnwidth]{978-3-903176-58-4~\copyright 2023 IFIP \hfill} \hspace{\columnsep}\makebox[\columnwidth]{ }}
\maketitle
 
\begin{abstract}
A popular branch of stochastic network calculus (SNC) utilizes moment-generating functions (MGFs) to characterize arrivals and services, which enables end-to-end performance analysis. However, existing traffic models for SNC cannot effectively represent the complicated nature of real-world network traffic such as dramatic burstiness. To conquer this challenge, we propose an adaptive spatial-temporal traffic model: \sys.   Specifically, we model the temporal on-off switching process as a dual Markovian arrival process (dMAP)  and   the arrivals during the on phases as an autoregressive hidden Markov model (AR-HMM).  The \sys model fits in with the MGF-SNC analysis framework, unifies various state-of-the-art arrival models, and matches real-world data more closely.
  We perform extensive experiments with real-world traces under different network topologies and utilization levels. Experimental results show that \sys significantly outperforms prevailing models in MGF-SNC. 
\end{abstract}

\input{sections/1.introduction.tex}

\input{sections/2.design.tex}

\input{sections/3.mgf_snc.tex}

\input{sections/4.applications.tex}

\input{sections/5.conclusion.tex}

\bibliographystyle{ieeetran}
\bibliography{reference}
\input{sections/6.appendix}

\input{sections/7.ethics}

\end{document}

%% file: sections/1.introduction.tex
\section{Introduction}\label{sec:intro}

Performance analysis, especially the evaluation of performance bound, is essential for supporting sufficient quality of service (QoS) for delay-sensitive network applications. Network calculus (NC) provides an ascendant methodology focusing on computing the performance bounds, i.e., the delay and backlog bounds. There exist two branches in NC: deterministic network calculus (DNC)~\cite{61109,61110,1134304} and stochastic network calculus (SNC)~\cite{Chang00,1638528,Fidler06,Jiang08,2342426,POLOCZEK201456,NIKOLAUS2019188,3388848,8264856}.    
DNC relies on deterministic models to \textcolor{black}{analyze} the strict worst-case performance, which usually results in quite low network utilization. However, most commercial network services allow statistical multiplexing of traffic flows to enhance network utilization, and prefer  probabilistic QoS metrics. Then SNC has been proposed by considering statistical behaviors to calculate the probabilistic performance bounds at a known small violation probability. Figure~\ref{fig:plane} shows the framework of MGF-based SNC. Network traffic characteristics are extracted by the feature extraction   module, and then are used by the performance evaluation module based on the MGF-SNC framework. The results of performance evaluation will be fed back to the center controller for further planning and optimization.


To investigate  traffic flows, the analysis of inter-arrival times (IATs) and packet sizes is the basic step. If they are found to follow certain probability distributions, the corresponding stochastic model of these flows can be determined. In MGF-SNC, packet sizes (or arrivals) and IATs are commonly assumed to conform to conventional distributions, including the Normal, Exponential,  et.al. However, our analysis and extensive past attempts~\cite{205464,1354569, LI20082584, LI2020123982,9500818,1064240} reveal that none of the conventional distributions are accurate enough for real network flows.

 \begin{figure}[tp]
 	\vspace{1em}
 	\centering
 	\includegraphics[width=0.48\textwidth]{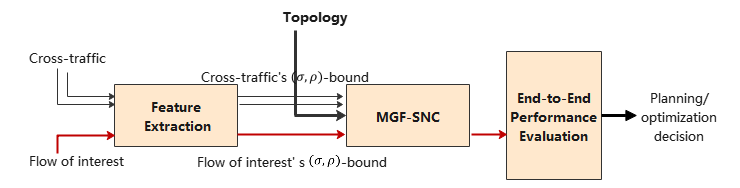}
 	\caption{\label{fig:plane} The workflow of  end-to-end  performance evaluation in MGF-SNC.}
 	\vspace{-1em}
 \end{figure} 
As early as the 1970's, ~\cite{1455418} reported a noticeable behavior of traffic, which is called ``burstiness" defined by peak to average transmission rate. It implies that an on-off switching pattern widely exists in real traffic flows.
The early literature quantitatively describes   real flows from a fractal perspective~\cite{205464,Li09,380206,392383} and has been subsequently applied by~\cite{LI2020123982, Li21, LI2021126138,LI20082584}.  The fractal models convincingly reveal traffic characteristics, such as self-similarity (SS) and long-range dependence (LRD). An important issue with them is that their MGFs grow   super-linearly with time and hence the underlying  SNC results for MGFs are not directly applicable~\cite{6145483}.    
 
To tackle the challenge in traffic modeling, we introduce a novel spatial-temporal model: \sys, as illustrated in Figure~\ref{dmapar-hmm}. Inspired by  the definition of modulation  in telecommunications, we regard the arrival process as a transmission signal that can be demodulated into a carrier signal and an input signal of positive impulses.  Specifically, the carrier signal is modeled by a dual Markovian arrival process (dMAP) to depict the (temporal) on-off switching process observed in real traffic traces, and the input signal is modeled by an autoregressive hidden Markov model (AR-HMM) for the representation of the (spatial) data amount during \textit{on} phases.  
We show that various   state-of-the-art (SOTA) arrival models used in SNC can be readily represented by the proposed \sys model.  With traffic traces from real networks, we demonstrate that \sys can accurately depict multi-dimensional multi-order traffic characteristics.    
 
For providing performance analysis, we develop the  $(\sigma(\theta), \rho(\theta))$-envelope  of the MGF of   \sys, which can fit in with the framework of MGF-SNC. Though here we focus on deriving the delay bound with \sys, the results can be extended without much effort to the backlog bound. Besides the $(\sigma(\theta), \rho(\theta))$-envelope of   MGF, more general envelopes, such as stochastically bounded burstiness envelope, could also be considered similarly~\cite{6868978}. By extensive experiments under   different network  topologies and utilization levels, we show that \sys can significantly boost the effectiveness of MGF-SNC, which means that both the tightness and the reliability of the delay bound are enhanced.

The rest of the paper is organized as follows. In Section~\ref{sec:design}, the spatial-temporal model \sys is introduced.  In Section~\ref{sec:mgf_snc}, we present the interface of the proposed traffic model to MGF-SNC for network traffic. In Section~\ref{sec:app}, the feasibility of \sys is examined with real-world traces. 
Section~\ref{sec:conclucion} concludes the article.

%% file: sections/2.design.tex
\section{Network Traffic}\label{sec:design}
 In this section, we introduce the proposed traffic model \sys in details by showing how it captures the temporal and spatial dynamics of real-world traffic. We also reveal that many famous models are correlated to  \sys,  and show that our model favors the retention of critical traffic features. 

\subsection{Traffic Modeling}
We divide the time into discrete intervals to conform to the SNC approach.  We use $A(s,t)$ for the amount of data traffic arriving in the interval $(s,t]$, i.e., 
$$
A(s,t)=\sum_{k=s+1}^ta_k,
$$
where we adopt the convention $A(t,t)=0$. For ease of expression, we use $A(t)$ to represent $A(0,t)$. In addition, we use $a_k$ to represent data traffic arriving during the $k$-th timeslot, that is,  $a_k=A(k-1,k)=A(k)-A(k-1)$.

In the real world, traffic flows exhibit on-off switching patterns. Moreover, arrivals of all timeslots  are difficult to be entirely modeled as a single stochastic process. Therefore, we are motivated to parse the traffic flow and build an appropriate model for each component.
In this paper, we present a discrete-time spatial-temporal model \sys for traffic flows delivered in real networks. As shown in Figure~\ref{dmapar-hmm}, a traffic flow is discretized into a time series $\{a_t\}, t=1,2,\cdots$, based on a constant timeslot interval $\Delta t$. The time series $\{a_t\}$ can be considered to be formed by modulating an input signal of positive pulses by a carrier signal with an on-off switching scheme and unit amplitude. The input signal, i.e., the spatial model, $\{y_1, y_2,\cdots\}$,  contains the amount of non-zero arrivals along $\{a_t\}$.  The carrier signal, i.e., the temporal model, represents the temporal feature by recording the  IATs  of both on and off phases, i.e., $\{\tau_1^{ \rm off}, \tau_2^{ \rm off},\cdots\}$ and $\{\tau_1^{ \rm  on}, \tau_2^{ \rm  on},\cdots\}$, which are supplementary to each other. We will subsequently discuss the temporal and spatial models in details. Table~\ref{tab:symbol} lists the frequently used symbols in this paper.  

\begin{table}
	\centering
	\begin{tabular}{|c|p{6cm}|}
		\hline
		\textbf{Symbol} & \multicolumn{1}{c|}{\textbf{Meaning}}\\ \hline
		$\mathcal{S}^{\rm on}$ & state space of ${\rm MAP}^{\rm on}$    \\
		$\mathcal{S}^{\rm off}$ & state space of ${\rm MAP}^{\rm off}$    \\
		$\mathcal{N} $ & state space of $X$    \\
		$\mathcal{S} $ & state space of $Z$    \\
		$Q$ & transition probability matrix   of the carrier signal in the embedded state space\\
		$P$ & transition probability matrix   of X\\
		$T$ & transition probability matrix   of Z\\
		$X$ & hidden state of the input signal \\
		$Z$ & hidden state of the \sys signal \\
		$\Delta t$ & timeslot interval for discretization\\
		$\circ$ & Hadamard product\\
		$\otimes$ & Kronecker product\\
		$I_n$ & Identity matrix of size $n$\\
		
		\hline
	\end{tabular}
	\caption{\label{tab:symbol} Frequently used symbols in this paper.}
\end{table}
 
\begin{figure}[]
	\centering
	\includegraphics[width=0.41\textwidth]{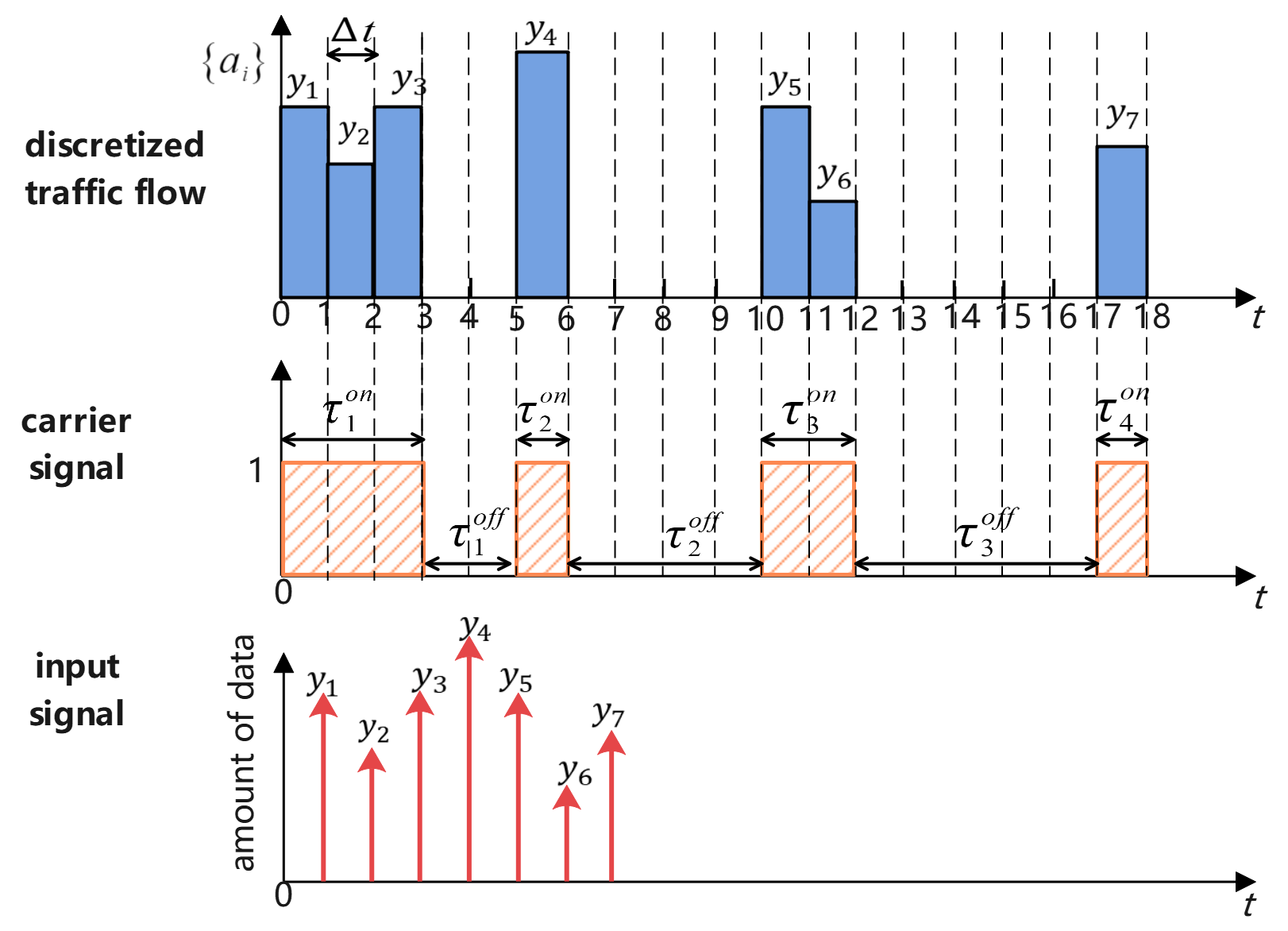}
	\caption{\bf \label{dmapar-hmm}\sys. { \rm Mapping the discrete-time observations into a spatial-temporal model, regarding the temporal on-off switching process as the carrier signal, and the data transmitted  in on phases as the input signal.
}}
\end{figure}

 \begin{table*}
	\centering
	\begin{tabular}{cp{13cm}}
		\toprule
		\textbf{Case} & \multicolumn{1}{c}{\textbf{1-step transition probability}}\\ \midrule
		\begin{minipage}{0.02\textwidth}
			\includegraphics[width=2.5mm, height=4.5mm]{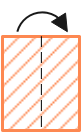}
		\end{minipage} & $(1,i,j)\curvearrowright(1,k,j)$: ${\rm MAP}^{\rm off}$ is frozen in state $j$, and ${\rm MAP}^{\rm on}$ makes a hidden transition from state $i$ to $k$. \\
		\begin{minipage}{0.02\textwidth}
			\includegraphics[width=2.5mm, height=4.5mm]{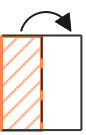}
		\end{minipage}  &$(1,i,j)\curvearrowright(0,i,k)$: ${\rm MAP}^{\rm on}$ is frozen in state $i$, and ${\rm MAP}^{\rm off}$ makes an observable transition from state $j$ to $k$.  \\ 
		\begin{minipage}{0.02\textwidth}
			\includegraphics[width=2.5mm, height=4.5mm]{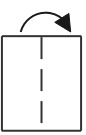}
		\end{minipage} & $(0,i,j)\curvearrowright(0,i,k)$: ${\rm MAP}^{\rm on}$ is frozen in state $i$, and ${\rm MAP}^{\rm off}$ makes a hidden transition from state $j$ to $k$.\\
		\begin{minipage}{0.02\textwidth}
			\includegraphics[width=2.5mm, height=4.5mm]{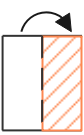}
		\end{minipage} & $(0,i,j)\curvearrowright(1,k,j)$: ${\rm MAP}^{\rm off}$ is frozen in state $j$, and ${\rm MAP}^{\rm on}$ makes an observable transition from state $i$ to $k$.\\ \bottomrule
	\end{tabular}
	\vspace{3pt}
	\caption{\label{tab:dmapinfo} State transition of the carrier signal. }
	
\end{table*}

 \begin{figure*}
 	\centering
 	\resizebox{\linewidth}{!}{
 		\begin{minipage}{22cm}
 			\begin{align}\label{eq: example Q}
 				Q=\scriptsize{
 					\begin{blockarray}{ccccc|cccc}
 						& (0,0,0) & (0,0,1) & (0,1,0)& (0,1,1) & (1,0,0) & (1,0,1) & (1,1,0)&(1,1,1)\\ 
 						\begin{block}{c[cccc|cccc]}				
 							(0,0,0) & 1-\nu_0\Delta t &\nu_{01}\Delta t & & & \frac{\nu_{02}^\prime (\nu_{02}+\nu_{03})\Delta t}{\nu^\prime_{02}+\nu_{03}^\prime}  & & \frac{\nu_{03}^\prime(\nu_{02}+\nu_{03})\Delta t}{\nu^\prime_{02}+\nu_{03}^\prime} &\\
 							(0,0,1) & \nu_{10}\Delta t  &1-\nu_1\Delta t &   & &    & \frac{\nu_{02}^\prime(\nu_{12}+\nu_{13})\Delta t}{\nu^\prime_{02}+\nu_{03}^\prime}  &   & \frac{\nu_{03}^\prime(\nu_{12}+\nu_{13})\Delta t}{\nu^\prime_{02}+\nu_{03}^\prime} \\
 							(0,1,0) &  &   & 1-\nu_0\Delta t  &\nu_{01}\Delta t & \frac{\nu_{12}^\prime(\nu_{02}+\nu_{03})\Delta t}{\nu^\prime_{12}+\nu_{13}^\prime}   &   & \frac{\nu_{13}^\prime(\nu_{02}+\nu_{03})\Delta t }{\nu^\prime_{12}+\nu_{13}^\prime}  &  \\ 
 							(0,1,1) & & & \nu_{10}\Delta t & 1-\nu_1\Delta t & &\frac{\nu_{12}^\prime(\nu_{12}+\nu_{13})\Delta t}{\nu^\prime_{12}+\nu_{13}^\prime}  & & \frac{\nu_{13}^\prime(\nu_{12}+\nu_{13})\Delta t}{\nu^\prime_{12}+\nu_{13}^\prime} \\
 							\cline{1-9} 
 							(1,0,0) & \frac{\nu_{02}(\nu^\prime_{02}+\nu^\prime_{03})\Delta t}{\nu_{02}+\nu_{03}}  &  \frac{\nu_{03} (\nu^\prime_{02}+\nu^\prime_{03})\Delta t}{\nu_{02}+\nu_{03}}   & & & 1-\nu^\prime_0\Delta t & &\nu^\prime_{01}\Delta t  & \\
 							(1,0,1) & \frac{\nu_{12}(\nu^\prime_{02}+\nu^\prime_{03})\Delta t}{\nu_{12}+\nu_{13}}  &  \frac{\nu_{13}(\nu^\prime_{02}+\nu^\prime_{03})\Delta t}{\nu_{12}+\nu_{13}}   & & & &1-\nu_0^\prime\Delta t &  &  \nu_{01}^\prime\Delta t \\
 							(1,1,0) & & &\frac{\nu_{02}(\nu^\prime_{12}+\nu^\prime_{13})\Delta t}{\nu_{02}+\nu_{03}}  &  \frac{\nu_{03}(\nu^\prime_{12}+\nu^\prime_{13})\Delta t }{\nu_{02}+\nu_{03}}   &   \nu_{10}^\prime\Delta t &   & 1-\nu_1^\prime\Delta t  &  \\ 
 							(1,1,1) & & &\frac{\nu_{12} (\nu^\prime_{12}+\nu^\prime_{13})\Delta t}{\nu_{12}+\nu_{13}}   &  \frac{\nu_{13} (\nu^\prime_{12}+\nu^\prime_{13})\Delta t}{\nu_{12}+\nu_{13}}  &  &\nu_{10}^\prime\Delta t & & 1-\nu_1^\prime\Delta t \\
 						\end{block}
 					\end{blockarray}
 				}
 			\end{align}
 		\end{minipage}
 	}
 \end{figure*}
 
\subsubsection{Temporal Model: dMAP}
 
 In general, the temporal feature of a flow  is quite challenging to characterize, and we propose a novel dMAP to enable an effective modeling methodology. The dMAP model consists of two independent continuous-time (or discrete-time) MAPs that randomly switch between each other. The two MAPs stand for the inter-arrival processes of the off and on phases, respectively. We choose MAPs for temporal modeling for two reasons. First,   MAP is regarded as one of the most expressive models for inter-arrival processes. Second, it is analytically tractable by using matrix-analytic methods. The dMAP model further enhances the representation capacity of the conventional MAP, and allows parallel computing for parameter estimation of two MAPs. 
 In our case, the carrier signal is controlled by  dMAP, where $\{\tau_1^{ \rm off}, \tau_2^{ \rm off},\cdots\}$ and $\{\tau_1^{ \rm  on}, \tau_2^{ \rm  on},\cdots\}$ are modeled by two independent MAPs -- ${\rm MAP}^{\rm off}(m_1)$ and ${\rm MAP}^{\rm on}(m_2)$ -- with a switching scheme.

 Let us denote MAP($m$) for an $m$-dimensional MAP which is governed by an underlying continuous-time Markov chain (CTMC) with state space $\mathcal{S}=\{0,\cdots, m-1\}, m\ge1$. Let $\mathbf{C}_0$ and $\mathbf{C}_1$ denote the matrices of state transition rates with zero and one arrival, respectively.     The following restrictions apply to the $\mathbf{C}_i$:
 $$
 \begin{array}{llll}
 	0&\le&[\mathbf{C}_1]_{ij} &<\infty\\
 		0&\le&[\mathbf{C}_0]_{ij} &<\infty\quad i\ne j\\
 		&&[\mathbf{C}_0]_{ii}&<0\\
 		\multicolumn{3}{c}{(\mathbf{C}_0+\mathbf{C}_1)\bm{1}}&=\bm{0},
 \end{array}
 $$
where $\bm{1}$ denotes an all-one column vector with an appropriate dimension.  In total, there are $2m^2-m$ free parameters.  The model parameters can be estimated using the   algorithms listed   in~\cite{9500818,10.1007}.
 To accommodate to the SNC framework,  we  shall discretize    ${\rm MAP}^{\rm off}(m_1)$ and ${\rm MAP}^{\rm on}(m_2)$       with      $\Delta t$. The timeslot interval  $\Delta t$ should be   small enough so that the probability of having two or more events occurring within $\Delta t$ can be neglected. Consider a continuous-time MAP(2) with the matrix representation
 $$
 \small{ 
 \mathbf{C}_0=\begin{bmatrix}
 	-\nu_0&\nu_{01}\\
 	\nu_{10}&-\nu_1
 \end{bmatrix}\quad {\rm and }\quad \mathbf{C}_1=\begin{bmatrix}
 	\nu_{02}&\nu_{03}\\
 	\nu_{12}&\nu_{13}
 \end{bmatrix},
}
 $$
  where $\nu_i$ is the sum rate at which the process makes a transition from state $i$, and $\nu_{ij}$ is the rate at which  the process makes a transition from state $i$ into state $j$. 
 The transition probability matrices of the corresponding discrete-time MAP are 
 $$
 \small{
 \mathbf{D}_0=\begin{bmatrix}
 	1-\nu_0\Delta t&\nu_{01}\Delta t\\
 	\nu_{10}\Delta t&1-\nu_1\Delta t
 \end{bmatrix}\;{\rm and }\;\;
 \mathbf{D}_1=\begin{bmatrix}
 	\nu_{02}\Delta t&\nu_{03}\Delta t\\
 	\nu_{12}\Delta t&\nu_{13}\Delta t
 \end{bmatrix},
}
 $$
and the row sum of $\mathbf{D}_0 + \mathbf{D}_1$ is always one.
For space considerations,   we  will  use  a set of model parameters    $\mathbf{D}=\{(\nu_i,\nu_{i,0},\cdots,\nu_{i,i-1}, \nu_{i, i+1}, \cdots, \nu_{i,2m-1})_{i=0,\cdots, m-1};    \Delta t\}$  to represent the discrete-time MAP.
 
\begin{table*}
	\centering
	\small{
		\begin{tabular}{l|r}
			\toprule
			\textbf{Model}  & \multicolumn{1}{|c}{\textbf{Architecture}}  \\   \hline
			MAP~\cite{9500818,10.1007,CASALE201061}& \mbox{disable ${\rm MAP}^{\rm on}$ +  AR-HMM($p=0, N=1, \mu=1,  \sigma= 0$)} \\
			BMAP~\cite{KLEMM2003149}& \mbox{disable ${\rm MAP}^{\rm on}$ +  AR-HMM($p=0, N=m_1\times\{\mbox{maximal batch size}\},   \sigma= 0$)}  \\
			Poisson~\cite{3388848, NIKOLAUS2019188}&\mbox{disable ${\rm MAP}^{\rm on}$ + ${\rm MAP}^{\rm off}$($m_1=1$) + AR-HMM($p=0, N=1, \mu=1, \sigma= 0$)}\\
			cPoisson & \mbox{disable ${\rm MAP}^{\rm on}$ + ${\rm MAP}^{\rm off}$($m_1=1$) + AR-HMM($p=0, N=1$)}\\
			MMOO~\cite{3388848, NIKOLAUS2019188}& \mbox{ ${\rm MAP}^{\rm off}$($m_1=1$) + ${\rm MAP}^{\rm on}$($m_2=1$) + AR-HMM($p=0, N=1,   \sigma=0$)}\\
			MMP~\cite{POLOCZEK201456} & \mbox{disable ${\rm MAP}^{\rm on}$ \& ${\rm MAP}^{\rm off}$ + AR-HMM($p=0$)}\\
			AR(p)~\cite{POLOCZEK201456} &\mbox{disable ${\rm MAP}^{\rm on}$ \& ${\rm MAP}^{\rm off}$ + AR-HMM($N=1$, residual=``normal")}\\
			Normal~\cite{POLOCZEK201456}&\mbox{disable ${\rm MAP}^{\rm on}$ \& ${\rm MAP}^{\rm off}$ + AR-HMM($p=0, N=1$, residual=``normal")}\\
			Exponential~\cite{3388848, NIKOLAUS2019188,POLOCZEK201456} & \mbox{disable ${\rm MAP}^{\rm on}$ \& ${\rm MAP}^{\rm off}$ + AR-HMM($p=0, N=1$,    residual=``exponential")} \\
			\bottomrule
		\end{tabular}
	}
	\vspace{3pt}
	\caption{\bf Popular arrival models widely used in MGF-SNC. \rm  The SOTA arrival models used in SNC can be readily represented by the proposed \sys model. }
	\label{tab:my_table}
\end{table*} 

 As aforementioned, the inter-arrival processes of the on and off phases are governed by ${\rm MAP}^{\rm off}(m_1)$ and ${\rm MAP}^{\rm on}(m_2)$ with  state spaces   $\mathcal{S}^{\rm off}=\{0,\cdots, m_1-1\}$ and $\mathcal{S}^{\rm on}=\{0,\cdots, m_2-1\}$, respectively. Let us use a tuple $(o,s^{\rm on}, s^{\rm off})$ to represent the state of the carrier signal at   timeslot $t$, with $o\in\{0,1\}$ indicating the physical state of the carrier signal (0: off and 1: on) and $(s^{\rm on}, s^{\rm off})\in\mathcal{S}^{\rm on}\times \mathcal{S}^{\rm off}$ the hidden state of the carrier signal, and 
   introduce a mapping $\iota_1$ to embed the higher-order state space into a 1-order one:
$$
\iota_1(o,s^{\rm on}, s^{\rm off})=o*m_1*m_2+s^{\rm on}*m_1+s^{\rm off}.
$$ 
 Given the matrix representation of ${\rm MAP}^{\rm on}$ and ${\rm MAP}^{\rm off}$,  the 1-step transition probability matrix $Q$ of the carrier signal   can be put   into a compact form, 
 $$
 Q=\left(\begin{array}{l|l}
 	Q_{\rm off\curvearrowright off}& Q_{\rm off\curvearrowright on}\\ \hline
 	Q_{\rm on \curvearrowright off} & Q_{\rm on\curvearrowright on}\\
 \end{array}
 \right),
 $$
 where $Q_{\rm off\curvearrowright off}$, $Q_{\rm off\curvearrowright on}$, $Q_{\rm on\curvearrowright off}$, and $Q_{\rm on\curvearrowright on}$ are four $(m1*m2)\times(m1*m2)$ block matrices, governing the on-off switching process of the carrier signal. All possible state transitions of $Q$ are summarized in Table~\ref{tab:dmapinfo}.  

 As an illustration, we show how to compute $Q$ in the following example. Suppose   we have ${\rm MAP}^{\rm off}(2)$ with $\mathbf{D}^{\rm off}=\{\nu_0, \nu_{01}, \nu_{02}, \nu_{03}; \nu_1, \nu_{10},\nu_{12}, \nu_{13}; \Delta t\}$ and 
${\rm MAP}^{\rm on}(2)$ with $\mathbf{D}^{\rm on}=\{\nu_0^\prime, \nu_{01}^\prime, \nu_{02}^\prime, \nu_{03}^\prime; \nu_1^\prime, \nu_{10}^\prime,\nu_{12}^\prime, \nu_{13}^\prime; \Delta t\}$. 
The  1-step transition probability matrix $Q$ is shown   in Eqn.~\eqref{eq: example Q}.

%
\begin{algorithm}[h] 
	\SetAlgoLined
	\SetKwProg{Fn}{def}{\string:}{}
	\SetKwProg{Rn}{return}{\string:}{}
	\SetKwInOut{Input}{Input}\SetKwInOut{Output}{Output}
	\algdef{S}[FOR]{ForEach}[1]{\algorithmicforeach\ #1\ \algorithmicdo}
	\caption{\label{alg:HMMFilters} Online EM algorithm for AR-HMMs.}
	\Input {\begin{itemize}
			\item Chain state space \footnotesize$\mathcal{N}=\{0,1,\cdots, N-1\}$\normalsize;
			\item Array of initial probabilities \footnotesize$\Pi_0=(\pi_0, \cdots,\pi_{N-1})\;$\normalsize  
			 such that  $\pi_i$  stores the probability that \footnotesize$X_0=i$\normalsize;
			\item Number of time lags  $p$ of the observable process;
			\item Sequence of observations \footnotesize$(y_0, y_1, \cdots )$\normalsize.
	\end{itemize} }
	\Output{The state transition matrix   \footnotesize$P$\normalsize,  and the model parameters \footnotesize$\Theta$\normalsize.}
	\BlankLine
	\ForEach {\footnotesize$time\; t  $\normalsize}{ 
		\tcc{E-step.}
		\ForEach {\footnotesize$state\; i,j \in \mathcal{N} $\normalsize}{
		\footnotesize	$\Xi_t\gets \mbox{Information matrix at time $t$}$\;
			\footnotesize$b_t\gets \Xi_t\Pi_{t-1}$\normalsize\;
			\footnotesize$Q\gets P^T$\normalsize\;
			\footnotesize$O[:, i]\gets Q\Xi_t O[:, i]+b_t[i]Q[:,i]$\normalsize\;
			\footnotesize$J[:,i,j]\gets Q\Xi_tJ[:, i,j]+b_t[i]Q[j,i]\bm{e}_j$\normalsize\;
			\ForEach {\footnotesize$ lag\;  l, r \in \{0,1,\cdots, p\}$\normalsize}{
				\footnotesize$F[:,i,l]\gets Q\Xi_t F[:,i, l]+b_t[i]Q[:,i]y_{t-l}$\normalsize\;
				\footnotesize$H[:,i,l,r]\gets Q\Xi_t H[:,i, l,r]+b_t[i]Q[:,i]y_{t-l}y_{t-r}$\normalsize\; }
		}
		\footnotesize$\Pi_t\gets Q b_t $\normalsize.
		\BlankLine
		\BlankLine
		\BlankLine
		\tcc{M-step.}
		\ForEach {\footnotesize$state\;  i,j\in \mathcal{N},\;\; lag\;  l \in \{1,\cdots, p\} $\normalsize}{
			\footnotesize$P[i,j]\gets \frac{J[i,j]}{O[i]}$\normalsize\;
			\footnotesize$\mu[i]\gets \frac{F[i,0]-\sum_{l=1}^p\phi[l,i]F[i,l]}{O[i]}$\normalsize\;
			{\footnotesize$\sigma[i]\gets \frac{\mu[i]^2O[i]+H[i,0,0] +  \sum_{l=1}^p\sum_{r=1}^p \phi[l,i]\phi[r,i]H[i,l,r]}{O[i]}$\normalsize}
			\footnotesize${\quad\quad\quad\;\;-\frac{2\mu[i]F[i,0]-2\sum_{l=1}^p\phi[l,i]H[i,0,l]}{O[i]}}$\normalsize\;
			\footnotesize$\phi[l,i]\gets\frac{H[i,0,l]-\sum_{r\neq l}\phi[r,i]H[i,l,r]}{H[i,l,l]}$\normalsize. 
		}
	}	 
\end{algorithm}
 
 \subsubsection{Spatial Model: AR-HMM}\label{ar-hmm} We investigate the models for characterizing the  spatial feature and find that the autoregressive hidden Markov model (AR-HMM)~\cite{hamilton1989new}  is suitable for this job.
 Suppose $X$ is a finite state homogeneous Markov chain.   Let  $\mathcal{N}=\{0,1,\cdots, N-1\}$ be the state space of $X$.  The observable process  $y$  has the form   
$$
	y_{t+1}=\mu(X_t)+\sum_{i=1}^{p}\phi_{i}(X_t)y_{t+1-i}+\sigma(X_t)\epsilon_{t+1}. 
$$
 Namely, the $(t+1)_{th}$ observation is influenced by the previous $p$ observations and by the state of $X$ at the previous time $t$ (i.e., the reaction to $X_t$ is not instantaneous).

The key components of an AR-HMM are: (1) $N$: the number of states of the hidden Markov chain; (2) $P$: the state transition matrix of the hidden Markov chain; (3) $p$: the number of time lags of the observable  process; and (4)  $\Theta=\{(\mu_0, \cdots, \mu_{N-1}), (\phi_{i,0}, \cdots, \phi_{i, N-1})_{i=1,\cdots,p},  (\sigma_0, \cdots, \sigma_{N-1}) \}$: the model parameters of the AR($p$)-HMM($N$). Suppose that the noise $\{\epsilon_t\}$ is a sequence of i.i.d. random variables (which are, in particular, independent of the $X_t$). Given $\{y_1, \cdots, y_t\}$,    $N$ and $p$ can be chosen according to the AIC/BIC information criteria~\cite{1100705, 10.2307/2958889}  as well as cross-validation,   and $P$ and $\Theta$ can be  estimated by using the EM-algorithm~\cite{7445144,ZHU2017223}. 
Define  
$$
\small{
	A(X_t)=\begin{bmatrix}
		\phi_1(X_t)&1&&\\
		\phi_2(X_t)&0&1&\\
		\vdots&&\ddots&\\
		\phi_p(X_t)&\cdots&&0
	\end{bmatrix},\quad
	C_{(t,t)}=\begin{bmatrix}
		\phi_1(X_t)\\
		\vdots\\
		\phi_p(X_t)
	\end{bmatrix} }
$$
and $C_{(s,t)}=A(X_s)C_{(s+1, t)}$, $\forall s<t$. 
 For  stationary purpose,   
$\lim_{n\to\infty}C_{(t-n,t)}=0.$
When $p=1$ and $N=1$, it is equivalent to  $|\phi|<1$.



\subsubsection{Observation Process: \sys}
Let us denote
$$
{\small  
	O_t=1_{\{\mbox{timeslot $t$ is in the on phase}\}}\text{ and }
	N(t)=\sum_{i\le t}O_i.
}
$$
The observation process $a$ can then be written as    $a_t=O_t\cdot y_{N(t)}, t=1,2,\cdots$.  The random variable $O_t$ indicates the physical state of the carrier signal at time $t$, and   $y_{N(t)}$ is  the amplitude of the input signal at time $t$.

  Let  $Z$ be a discrete-time Markov chain whose state space is 
$\mathcal{S}=\{\iota(o,s^{\rm on}, s^{\rm off}, i)| (o, s^{\rm on}, s^{\rm off}, i)\in\{0,1\}\times \mathcal{S}^{\rm on}\times\mathcal{S}^{\rm off}\times\mathcal{N}\}$, where $\iota(o,s^{\rm on}, s^{\rm off}, i)=N*\iota_1(o,s^{\rm on}, s^{\rm off})+i$. The transition probability matrix of $Z$ is given by 
 \begin{equation}\label{1step_st}
	 T=\left(
 \begin{array}{l|l}
 	Q_{\rm  off\curvearrowright off}\otimes I_N& Q_{\rm off\curvearrowright on}\otimes P \\\hline
 	Q_{\rm on\curvearrowright off}\otimes I_N& Q_{\rm on \curvearrowright on }\otimes  P
 \end{array}
 \right).
 \end{equation}
Thereby $a_t$   can  be rewritten as
$$
\begin{array}{ll} 
	 a_t=&\displaystyle \breve\mu(Z_{t-1}, Z_t) + \displaystyle \sum_{i=1}^p\breve\phi_i(Z_{t-1},Z_t)y_{N(t-1)+1-i}\\
	 & + \displaystyle \breve\sigma(Z_{t-1}, Z_t)\epsilon_t,
 \end{array}
$$
where  
$\breve{g}(Z_{t-1}, Z_{t})=1_{\{Z_t\ge m_1*m_2*N\}}g_{Z_{t-1}\%N}$, 
$g\in\{\mu, \phi_1, \cdots, \phi_p, \sigma\}$, and  $\%$  is  the modulo division.   

It is worth noticing that many well-known arrival models, such as MAP, batch MAP (BMAP)  et.al.,  are  special cases of the proposed \sys model (see Table~\ref{tab:my_table} for more details). Another benefit of \sys is that it is  additive when the residuals conform to  additive white Gaussian noise. Namely, if each individual flow is fitted by \sys,  the aggregate flow can also be modeled  by \sys.    In view of its strong generalization, \sys can be used to model various real-world traffic flows, which is, however, a great challenge to many existing models.
%

 \begin{figure}
 	\centering
 	\begin{subfigure}[b]{.85\linewidth}
 		\includegraphics[width=\linewidth]{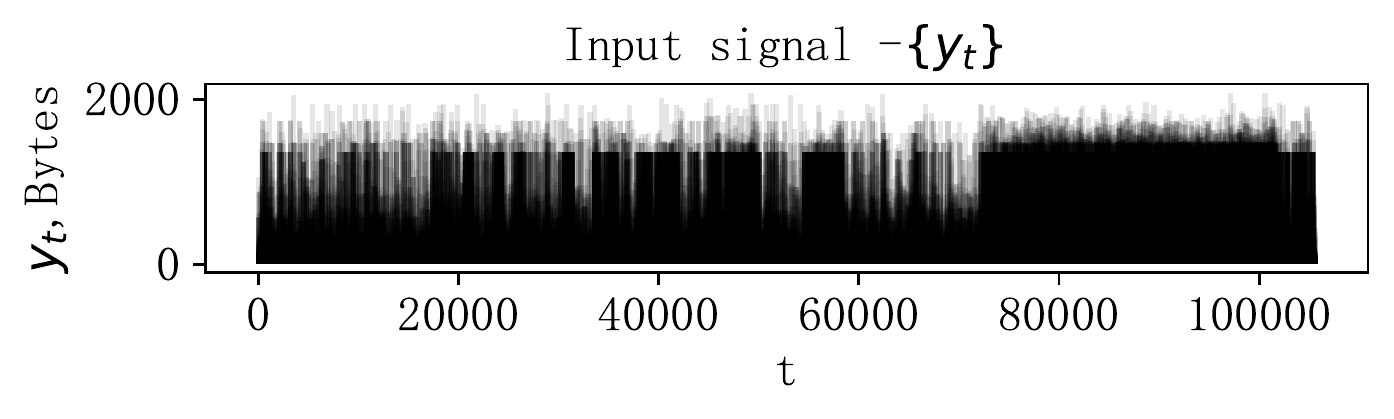}
 		\caption{Spatial feature of the discretized WIDE trace with (src IP, dst IP, protocol)=(203.76.247.112, 204.57.206.30, IPv4) captured at 2022-04-13 08:00. The x-coordinate  $t$ is the index of the time series $\{y_t\}$.} 
 	\end{subfigure}
 	
 	\begin{subfigure}[b]{.45\linewidth}
 		\includegraphics[width=\linewidth]{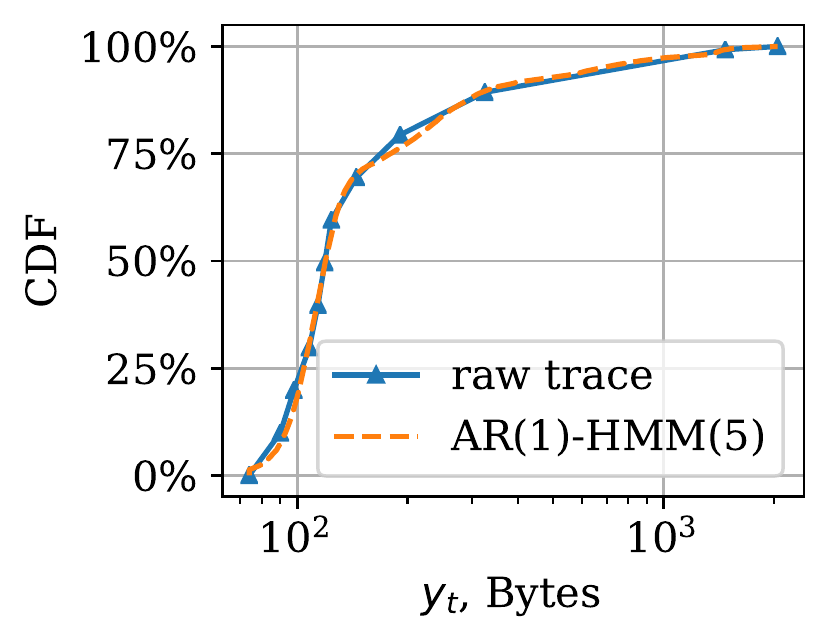}
 		\caption{CDF of $y_t$, $t\in[0, 1e4]$.} 
 	\end{subfigure}
 	\begin{subfigure}[b]{.45\linewidth}
 		\includegraphics[width=\linewidth]{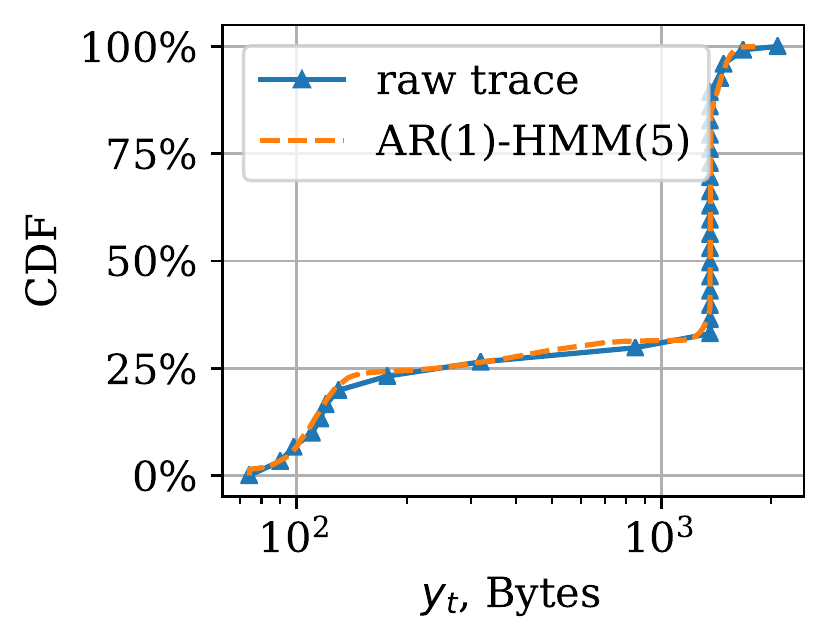}
 		\caption{CDF of $y_t$, $t\in[9e4, 1e5]$.} 
 	\end{subfigure}
 	\caption{{\bf Spatial feature extraction.} AR-HMM can capture the dynamics of the input signal.} 
 	\label{fig:sptiofet}
 \end{figure}

\subsection{The Applications in Traffic Feature Extraction}\label{tfe}
To verify the effectiveness of \sys  in real-world traffic feature extraction, we have conducted experiments on different kinds of traces, such as the multimedia streaming and video conference collected from Huawei commercial and campus networks, as well as open-source WIDE traces from MAWI\footnote{The dataset is available at \href{http://mawi.wide.ad.jp/}{http://mawi.wide.ad.jp/}.}, which is a mixture of commodity applications and research experiments. Due to space limitations and data openness, herein we show the results with typical traffic flows from the WIDE traces.


%
%

To accommodate the SNC framework, 
we adopt a constant timeslot interval of $\frac{mean(IATs)}{40}$ to discretize the time.  We  fit the discrete-time processes
with \sys and generate synthetic traces
using the well-fitted model. Specifically,  we use  Algorithm~\ref{alg:HMMFilters} to fit the spatial model, and    the algorithms listed in~\cite{9500818,10.1007} to fit the temporal model.  As we can see from   Figures~\ref{fig:sptiofet} and~\ref{fig:temporalfet},  the proposed   model can well capture the dynamics of   the input and the carrier signals. 
   We   compare the real and synthetic traces 
by showing the sample path as well as the accumulated
arrivals in Figure~\ref{fig:features}, which indicates that our model can efficiently learn the arrival patterns from the input traffic. 
 We  further  compare these traces in terms of the Hurst exponent H and the coefficient of variation CV (see Table~\ref{tab:comparison} for more details). It can be concluded from  the experimental results   that the proposed   model can exactly capture the various characteristics of traffic  from the decomposed spatio-temporal signals.   
 \begin{figure}
 	\centering
 	\begin{subfigure}[b]{.85\linewidth}
 		\includegraphics[width=\linewidth]{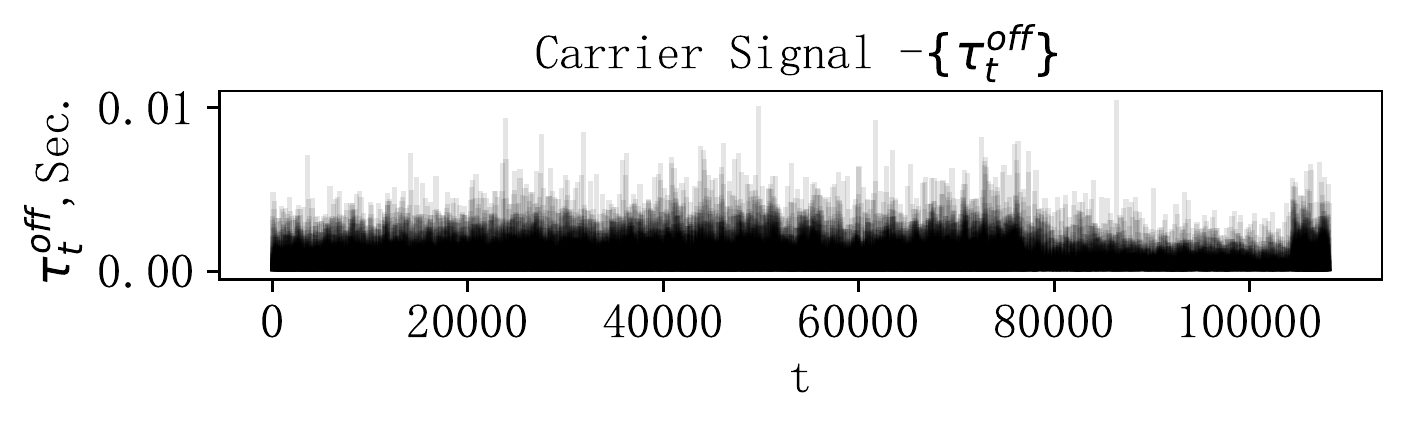}
 		\caption{Temporal feature of the discretized WIDE trace as  shown in Figure~\ref{fig:sptiofet}.  The x-coordinate  $t$ is the index of the time series $\{\tau_t^{\rm off}\}$.\\} 
 	\end{subfigure}      
	\begin{subfigure}[b]{.45\linewidth}
	\includegraphics[width=\linewidth]{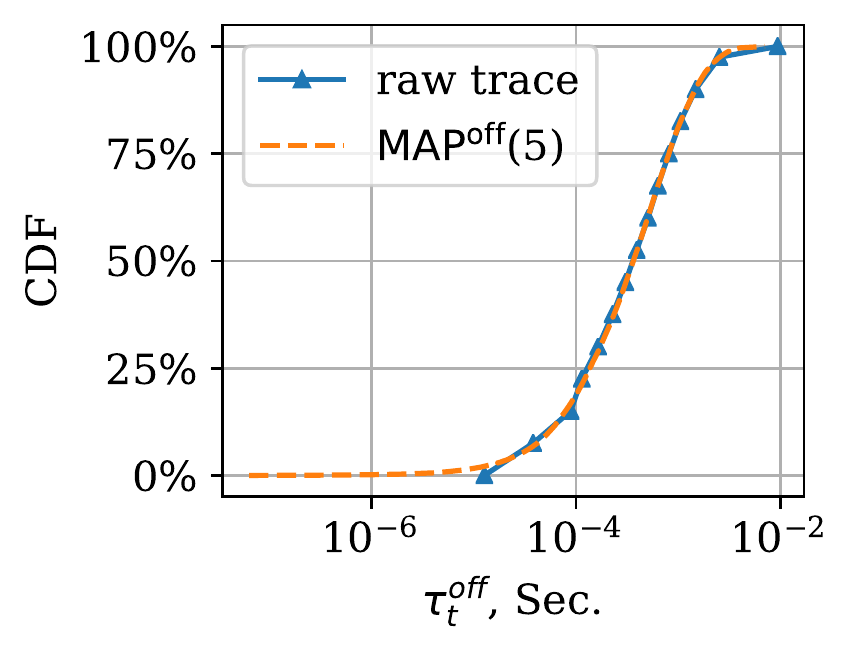}
	\caption{CDF of $\tau_t^{\rm off}$, $t\in[0, 1e4]$} 
\end{subfigure}
\begin{subfigure}[b]{.45\linewidth}
	\includegraphics[width=\linewidth]{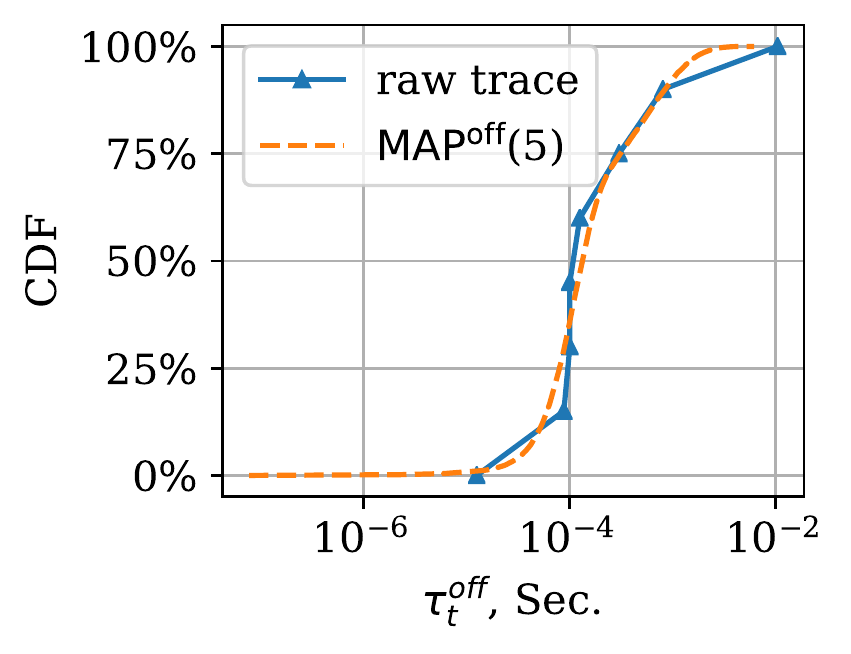}
	\caption{CDF of $\tau_t^{\rm off}$, $t\in[9e4, 1e5]$} 
\end{subfigure}
\caption{{\bf Temporal feature extraction.}  MAP can capture the dynamics of the carrier signal.}
\label{fig:temporalfet}
\end{figure}

 \begin{figure*}
 	\setlength{\abovecaptionskip}{-0.05pt}
 	\centering 
 	\begin{subfigure}[b]{.3\linewidth}
 		\centering
 		\includegraphics[width=0.95\textwidth]{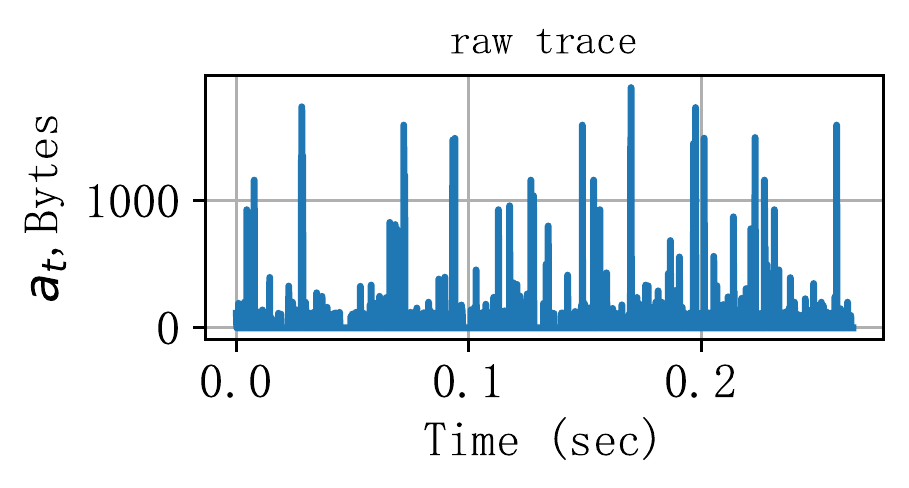}
 		\caption{ }  
 	\end{subfigure}
 	\begin{subfigure}[b]{.3\linewidth}
 		\centering
 		\includegraphics[width=0.95\textwidth]{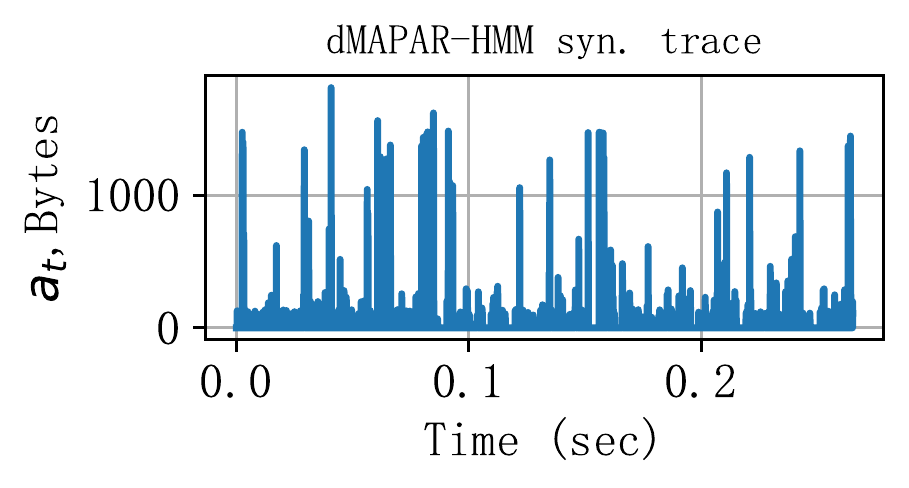}
 		\caption{ }  
 	\end{subfigure}
 	\begin{subfigure}[b]{.28\linewidth}
 		\centering
 		\includegraphics[width=1\textwidth]{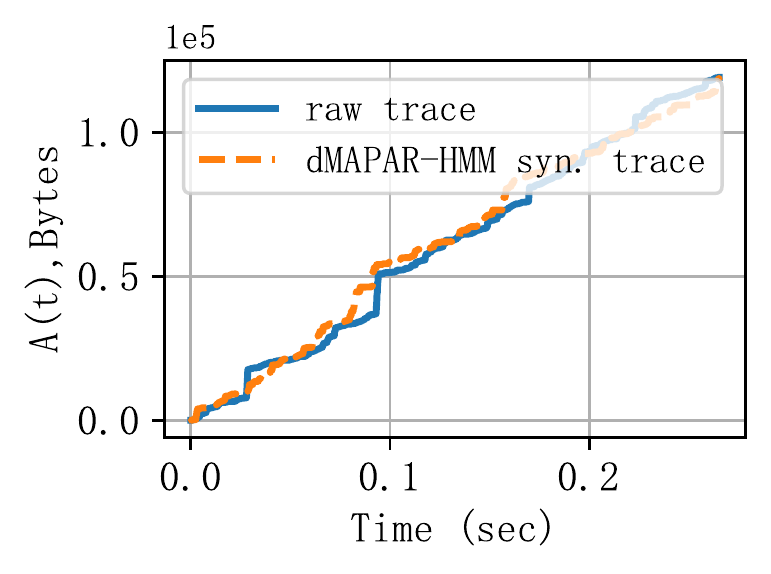}
 		\caption{ }  
 	\end{subfigure}
 	\vspace{1em}
 	\caption{\label{fig:features} {\bf Sample paths and cumulative arrivals of raw and synthetic traces.}  {\rm The figures are, from left to right: (a) discretized WIDE trace; (b) synthetic trace with   well estimated \sys; and (c) accumulative arrival (bytes) over a 0.25-second window. }}
 \end{figure*}		 
 
\begin{table*} 
	\begin{center}
		\footnotesize{
			\begin{tabular}{ |lccccc| } 
				\hline
				WIDE Trace captured at 2022-04-13 08:00   & FlowID      & H ({\it empirical}) & H ({\it estimated}) & CV ({\it empirical}) & CV  ({\it estimated})  \\ \hline
			  (src IP, dst IP, protocol)=(203.76.247.112, 204.57.206.30, IPv4) & flow1 & 0.2977 & 0.2978 &11.43&11.18 \\
			  (src IP, dst IP, protocol)=(13.212.219.174, 163.194.192.18, UDP) & flow2  & 0.2739 & 0.2790 &21.59&22.27\\
			  (src IP, dst IP, protocol)=(131.142.41.211, 144.195.33.147, TCP) & flow3  & 0.3085  & 0.3055  &27.04 & 26.20\\
				\hline
		\end{tabular}}
		\vspace{3pt}
		\caption{\label{tab:comparison} {\rm Comparison between raw and synthetic traces. }
		}
	\end{center}
\end{table*}

%
%
%

%% file: sections/3.mgf_snc.tex
\section{Interface to MGF-SNC}\label{sec:mgf_snc}

In this section, we   deduce an interface for the spatio-temporal model   to   MGF-SNC. 

An arrival process $A(s,t)$ is defined as being $(\sigma_A, \rho_A)$-constrained if for all $\theta>0$, there exist  $\sigma_A(\theta)$ and  $\rho_A(\theta)\in \mathbb{R}_+\cup\{+\infty\}$ such that the MGF of $A(s,t)$ satisfies
$$
\mathbb{E}[e^{\theta A(s,t)}]\le e^{\theta(\sigma_A(\theta)+\rho_A(\theta)(t-s))}.
$$ 
This is related to the theory of effective bandwidth that is defined as
\begin{equation}\label{equ:effectivebd}
\frac{1}{\theta(t-s)}\ln \mathbb{E}[e^{\theta A(s,t)}].
\end{equation}
The effective bandwidth characterizes the statistical behavior of traffic arrivals and has be applied to the derivation of closed-form solutions for end-to-end performance bounds~\cite{Fidler06}.

\subsection{MGF of \sys}
Consider an arrival process   with \sys arrivals. Let $V(s,t)$ be the vector that stands for $[\mathbb{E}[e^{\theta A(s,t)}|Z_s=k]]_{k\in\mathcal{S}}$,  which is the MGF of $A(s, t)$ conditional on the state of $Z_s$ at time $s$. By definition, we have $V_k(t,t)=1, \forall k\in\mathcal{S}$. For  all $s\le t$, we have the following result for the MGF.

\begin{theorem}\label{cmgf}
 For all $s\le t$, we have 
$$
\begin{array}{lll}
V_k(s,t)&=&\displaystyle \mathbb{E}\Big[e^{\sum_{i=1}^p[\theta_i(s,t)-\theta] y_{N(s)+1-i}}\\
&&\displaystyle \;\;\;\;\times\prod_{l=s+1}^tM(\theta_1(l,t); Z_{l-1}, Z_l)|Z_s=k\Big],
\end{array}
$$
	where   $\theta_i(s,t)=\theta\varphi_i(s,t)$, with 
	\begin{icompact}
		\item $\varphi_{p+1}(s,t)=1$ 
		\item $\varphi_i(t,t)= 1$
		\item 
		$\varphi_i(s-1,t)=  (1-O_s)\varphi_i(s, t)+O_s[\varphi_1(s,t)\breve{\phi}_i(Z_{s-1}, Z_s)+\varphi_{i+1}(s,t)]$ 
	 
	\end{icompact}
      $\forall i=1, \cdots, p$   and $M(\cdot; Z_{l-1}, Z_l)$ is the MGF of $\breve{\mu}(Z_{l-1}, Z_l)+\breve{\sigma}(Z_{l-1}, Z_l)\epsilon_l$ given $Z_{l-1}, Z_l$.  
\end{theorem} 
\begin{proof}
The theorem can be easily proved by induction. Interest readers can refer to the appendix  for the details.
\end{proof}
 
 \subsection{  Asymptotic Effective Bandwidth}
 
As to Theorem~\ref{cmgf}, we first note  that  $O_s=1_{\{Z_s\ge m_1*m_2*N\}}.$  If we let      $\varphi(s,t)=[\varphi_1(s,t), \cdots, \varphi_p(s,t)]^T$, $b(Z_s)=[0,0,\cdots,O_s]^T$   and   
$$
 \scriptsize{
 	A(Z_{s-1},Z_s)=
 \begin{bmatrix}
	1-O_s+ \breve{\phi}_1(Z_{s-1}, Z_s)& O_s&\\
	\breve{\phi}_2(Z_{s-1}, Z_s)&1-O_s&O_s&\\
	\vdots& &\ddots\\
	\breve{\phi}_p(Z_{s-1}, Z_s)&\cdots&\cdots&1-O_s
\end{bmatrix},
}
$$
 we would  have 
\begin{equation}\label{recursive}
	\small{
\begin{array}{ll}
	 &\varphi(s-1,t) \\
	=& b(Z_s)+A(Z_{s-1}, Z_s) \varphi(s, t)\\
    =&b(Z_s)+A(Z_{s-1}, Z_s)[b(Z_{s+1})+A(Z_{s}, Z_{s+1})\varphi(s+1, t)]\\
	=&\cdots\\
	=&\varphi(s-1,t-1)+  A(Z_{s-1}, Z_s)\cdots A(Z_{t-2}, Z_{t-1})\\
	&\quad\times [b(Z_t)+A(Z_{t-1}, Z_t)\bm{1}-\bm{1}]\\
	\to& \varphi(s-1,t-1),\quad \mbox{as $t-s\to \infty$}.
\end{array}
}
\end{equation}

The last step of Eqn.~(\ref{recursive}) used
 the stationary condition of AR-HMM, which implies that  $\{\varphi(s,t)\}_{t}$ converges as $t$ tends to infinity.    
Let     
$
 \lim_{t-s\to \infty}\varphi(s,t)=W_s,
$
the following results hold for the effective bandwidth. 
  
 \begin{theorem}\label{clm}
 	Let    $\bm{v}=[\mathbb{E}[W_s|Z_s=k]]_{k\in\mathcal{S}}$,  $A$ be the matrix $[A(Z_{s-1}=i, Z_s=j)]_{i,j\in \mathcal{S}}$ with its element belonging to $\mathbb{R}^{p\times p}$, and $B=\left[
	[0, \cdots, 0]_{p\times 1},
	\cdots,
	[0, \cdots, 0]_{p\times 1},\right.$ $\left.
	[0, \cdots, 1]_{p\times 1},
	\cdots,
	[0, \cdots, 1]_{p\times 1}
	\right]^T $.
%
%
%
We then have 
\begin{enumerate}
	\item
  $\bm{v}=  (I_{|\mathcal{S}|\times p}-T\circ A)^{-1}T B$;    
	\item  Denote $\phi_{\max}=\max   |\phi_{i,j}|$  and let  $\bm{\pi}=[\pi_i]_{i\in \mathcal{S}}$ be the stationary probability vector of $Z$, where $\pi_i$ stands for the stationary probability of state $i$. 
	We then have
	$$
	\small{
	\begin{array}{lll}
	\displaystyle \lim_{\phi_{\max}\to 0}1/\theta(t-s)(\ln \mathbb{E}[e^{\theta A(s,t)}]- 
	\displaystyle \ln \bm{\pi}{\rm R} \Gamma^{t-s}(\theta)\bm{1})=0,
	\end{array}}
	$$
where
$$
\small{  
\left\{  
\begin{array}{lll}
	{\rm R}={\rm diag}([\mathbb{E}[e^{ \sum_{k=1}^p\theta(v_{i,k}-1) y_{N(s)+1-k}}|Z_s=i]]_{i\in\mathcal{S}})\\
	\Gamma(\theta)=[\psi_{i,j}(\theta)]_{i,j\in\mathcal{S}}\circ T\\
	\psi_{i,j}(\theta)=\left\{
	\begin{array}{ll}
		\displaystyle e^{\theta\breve{\mu}_{i,j} v_{j,1}+\frac{1}{2}\theta^2\breve{\sigma}_{i,j}^2 v_{j,1}^2},&\epsilon_t\sim\mathcal{N}(0,1),\\
		\displaystyle \frac{e^{\theta\breve{\mu}_{i,j} v_{j,1}}}{1-\theta\breve{\sigma}_{i,j} v_{j,1}},&\epsilon_t\sim {\rm Exp}(1).
	\end{array}
	\right. 
\end{array}
\right.
}
$$
\end{enumerate}
 
 \end{theorem}
 \begin{proof}
 The proof can be found in the appendix. 
\end{proof}

$$\quad$$

Based on Theorem~\ref{clm} and Eqn.~\ref{equ:effectivebd}, we can obtain the asymptotic effective bandwidth. Since $\Gamma(\theta)$ is diagonalizable in  $\mathbb{C}$, there exists a diagonal matrix $D$ and a matrix $U$ such that $\Gamma(\theta)=UDU^{-1}$. Hence,
$$
\begin{array}{ll}
	\bm{\pi}{\rm R}\Gamma^{t-s}(\theta)\bm{1}&=\bm{\pi}{\rm R}UD^{t-s}U^{-1}\bm{1}\\
	&=\sum u_i D_{ii}^{t-s}\\
	&=\max |D_{ii}|^{t-s}\sum u_i (\frac{D_{ii}}{\max |D_{ii}|})^{t-s}\\
	&\le \max|D_{ii}|^{t-s}\sum u_i\\
	&=e^{\ln\sum u_i+\ln\max|D_{ii}|(t-s)},
\end{array}
$$
with the notation $u_i=(\bm{\pi}{\rm R}U)_i(U^{-1}\bm{1})_i$.

  \begin{algorithm}[t] 
  	\SetAlgoLined
  	\SetKwProg{Fn}{def}{\string:}{}
  	\SetKwProg{Rn}{return}{\string:}{}
  	\SetKwInOut{Input}{Input}\SetKwInOut{Output}{Output}
  	\SetKwFunction{modelest}{FeatureExtractor}
  	\SetKwFunction{aefbd}{effective$\_$bandwidth}
  	\caption{\label{alg:algo}  Feature Extraction}
  	\Input {flow, $\Delta t$, $\theta$}
  	\Output{$(\sigma, \rho)$-bound}
  	\BlankLine
  	\Fn(){\modelest{flow, $\Delta t$, $\theta$}}{
  		\tcc{from continuous-time to discrete-time system.}
  		$\{a_t\} \gets$  discretize the    traffic flow with $\Delta t$ \;
  		$(Q, P,\Theta) \gets$  fit 	\sys  with $\{a_t\}$\;

  		\tcc{calculate the input flow's $(\sigma(\theta), \rho(\theta))$-bound with the   extracted features.}
  			$T\gets$   Eqn.(\ref{1step_st})\; 
  			$\bm\pi\gets \bm\pi T=\bm\pi, \sum\pi_i=1$\;
  		$R(\theta)\gets$   Theorem~\ref{clm} \;
  		$\Gamma(\theta)\gets$  Theorem~\ref{clm} \;
  		 $(\sigma_A(\theta), \rho_A(\theta)) \gets$   Eqn.~\eqref{efbd}.\\
  		\textbf{return}	 $(\sigma_A(\theta), \rho_A(\theta))$
  	} 
  \end{algorithm}

The asymptotic effective bandwidth can then be programmatically calculated with
\begin{equation}\label{efbd}
\begin{array}{lll}
 \frac{1}{\theta(t-s)}\ln \mathbb{E}[e^{\theta A(s,t)}]& \sim &\frac{1}{\theta(t-s)}\ln\bm{\pi}{\rm R}\Gamma^{t-s}(\theta)\bm{1}\\
&\le & \sigma_A(\theta)/t-s+\rho_A(\theta),
\end{array}
\end{equation}
where $\sigma_A(\theta)= \ln\sum u_i/\theta$ and  $\rho_A(\theta)=\frac{\ln\max|D_{ii}|}{\theta}$.

%% file: sections/4.applications.tex
\section{End-to-End Performance Evaluation}\label{sec:app}
 
We have established a spatial-temporal model   for network traffic, and introduced a ready-to-use interface for it to the MGF-SNC in previous sections. In this section, we shall use it  directly to conduct end-to-end network analysis. 
Alike arrival process, a service process $S(s,t)$ is said to be   $(\sigma_S, \rho_S)$-constrained if for all $\theta>0$, there exist  $\sigma_S(-\theta)$ and  $\rho_S(-\theta)\in \mathbb{R}_+\cup\{+\infty\}$ such that 
$$
\mathbb{E}[e^{-\theta S(s,t)}]\le e^{\theta(\sigma_S(-\theta)-\rho_S(-\theta)(t-s))}.
$$ 

There are three atomic operations underlying the MGF-SNC analysis framework~\cite{Chang2000}: \textit{Aggregation}, \textit{Left-over} and \textit{Concatenation}.  With these tools, we can reduce any feed-forward   to a single flow - single server topology.     From SNC theory, we know that, if arrivals and service are $(\sigma, \rho)$-constrained, 
the stochastic delay bound which is violated at most with probability $\varepsilon$ is given by~\cite{6868978}  
\begin{equation}\label{delaybound}
\begin{array}{lll}
	\displaystyle T_\varepsilon=\inf_{\theta\in\{\theta|\rho_A(\theta)<\rho_S(-\theta)\}}\Big\{\frac{\sigma_A(\theta)+\sigma_S(-\theta)}{\rho_S(-\theta)}\\
	\displaystyle \quad\quad\quad\quad\quad -\frac{\log( \varepsilon(1-e^{ \theta(\rho_A(\theta)-\rho_S(-\theta))}))}{\theta\rho_S(-\theta)}\Big\}.
\end{array}
\end{equation}

\subsection{Experimental Setup}

Let us consider the topology shown in  Figure~\ref{fig:topology}.  The  switches support  two queues. Flow 1 enters the highest priority queue. Flow 2 and 3 enter the lowest priority queue. The queuing mode is  by Strict Priority where the priority sets the order in which queues are served, starting with the highest priority queue and going to the next lower queue only after the highest queue has been transmitted.  Suppose that switches in this network are  work-conserving   servers with constant service rate $c$, i.e., $\sigma_S(-\theta)=0$ and $\rho_S(-\theta)=c$.  For any flow of interest,  we can  use Eqn.~(\ref{delaybound}) to compute the end-to-end  performance bounds with $S$ being  substituted  by the end-to-end service $S_{e2e}$. 
  To illustrate the importance of   feature extraction in performance analysis, we take 6 well-known arrival models to benchmark the delay bound, which are: (1) independent normally distributed increments (Normal); (2) independent exponentially distributed increments (Exponential); (3) compounded Poisson (cPoisson); (4) 1-order autoregressive process (AR1); (5) Markov-modulated on-off (mmoo); and (6) Markov-modulated process (mmp). 
  Comparisons are made against   the ground truth  which  is   obtained from the discrete-event simulation (DES)~\cite{Li_ns_py_2022}.

 \begin{figure}[tp]
	\vspace{1em}
	\centering
	\includegraphics[width=0.38\textwidth]{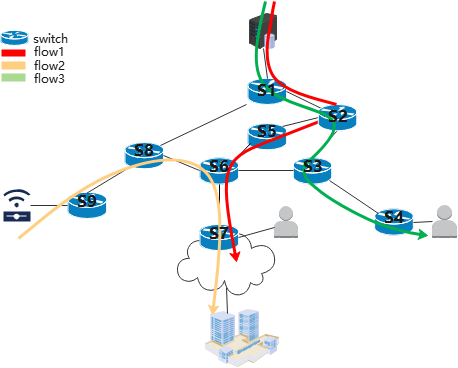}
	\caption{\label{fig:topology} {\bf Network topology of a small-scale commercial campus.}  The flows are of typical backbone network traffic as shown in Table~\ref{tab:comparison}.  }
	\vspace{-1em}
\end{figure}  
\begin{figure}[t]
	\centering 
	\includegraphics[width=0.45\textwidth]{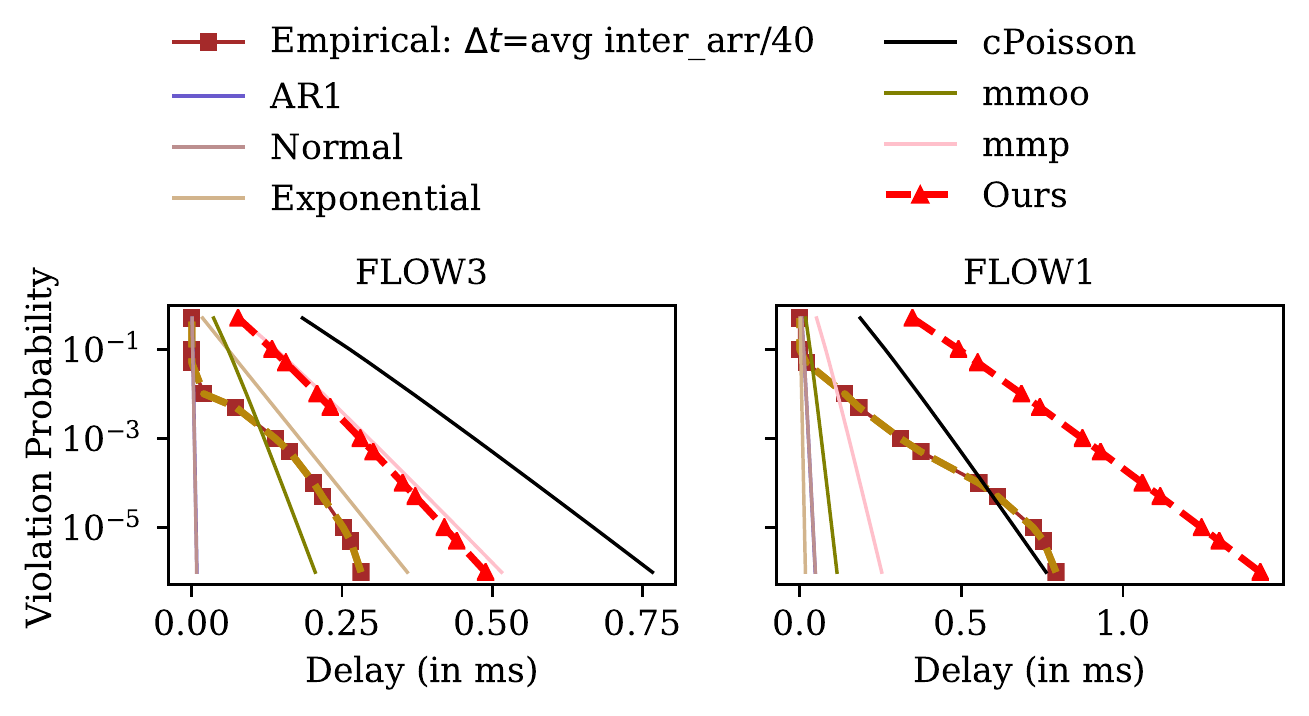}
	\caption{\label{fig:ss_results}  \rm  Performance evaluation  results on the single flow - single server topology with $c=100Mbps$.  }
\end{figure}
\subsection{Single Server Analysis}
  
We begin our analysis with the single flow - single server topology.  This is the most important topology, since we can reduce any feed-forward to this network by using the end-to-end service.  
Another benefit   of  it is that  we can eliminate interference from other sources and focus on the affect of arrivals on performance evaluation.  In Figure~\ref{fig:ss_results}, we test the impact of   arrival models    in feature extraction.     As we can see from this figure,  if we use the Exponential model to extract     traffic features in MGF-SNC,  we will get a good  estimate of the delay bound with flow 3.  But when we move over to flow 1, we cannot even get a reliable estimate. The same argument  applies to the mmp model and the cPoisson model.
  

\subsection{Network Analysis}
  \begin{figure*}
	\setlength{\abovecaptionskip}{-0.05pt}
	\centering 
	\begin{subfigure}[b]{0.75\linewidth}
		\includegraphics[width=0.9\textwidth]{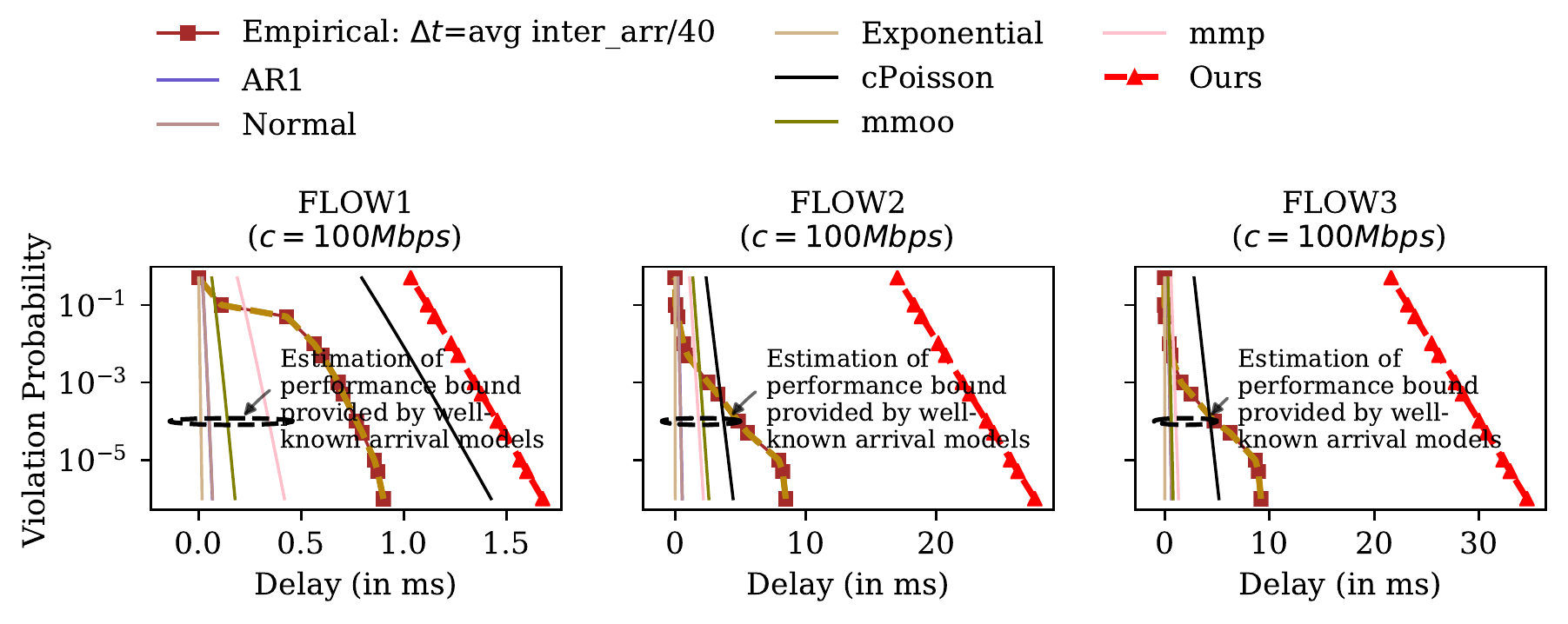}
		\includegraphics[width=0.9\textwidth]{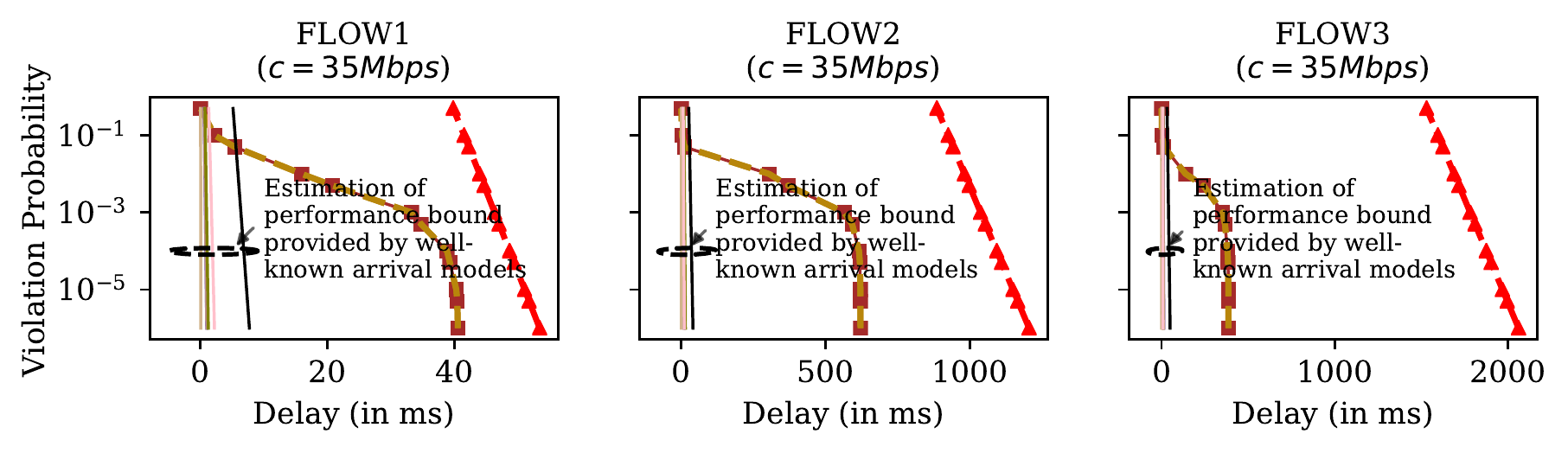}
		\caption{\label{fig:delaybound}}
	\end{subfigure}
	\begin{subfigure}[b]{0.24\linewidth}
		\includegraphics[width=1\textwidth]{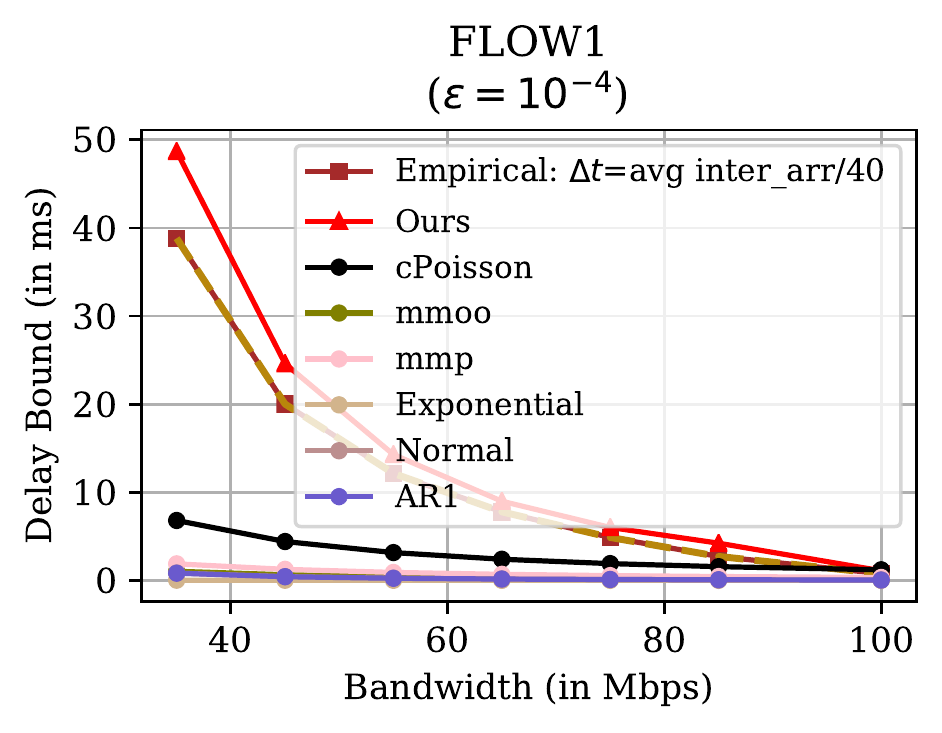}
		\vspace{1em}
		\includegraphics[width=1\textwidth]{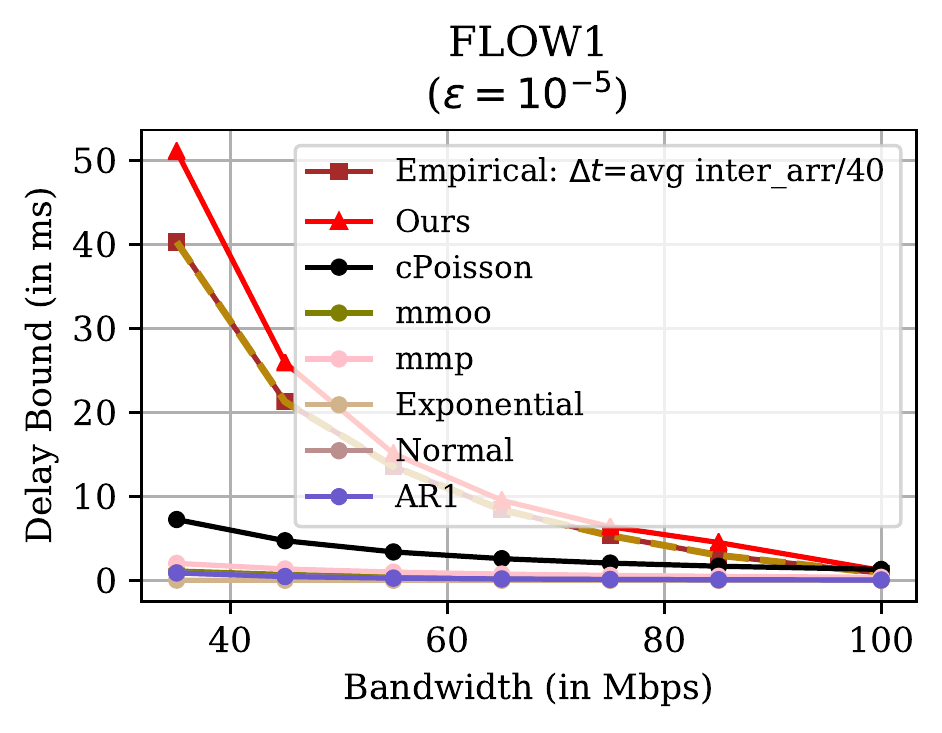}
		\caption{\label{fig:bandwidth}}
	\end{subfigure}
	\vspace{1em}
	\caption{\label{fig:comp_bm_des}(a) \rm  End-to-end performance analysis of the network topology; and (b) \rm  stochastic delay bound  with violation probability $\varepsilon=10^{-4}$ and $\varepsilon=10^{-5}$, respectively.}
\end{figure*}
We continue  with the network topology     and adopt the ``pay multiplexing only  once" (PMOO)~\cite{8264856} algorithm to perform the end-to-end network  analysis.    Let us consider flow 1. The overall service offered to  flow 1 can be described by the end-to-end service 
$$
S_{e2e}(s,t)=S_1\otimes S_2\otimes S_5 \otimes S_6 \otimes S_7  (s,t),
$$
where $\otimes$ is the {\it min-plus} convolution operator. The end-to-end service   is not influenced by the cross-flows, and the only uncertainty to $T_\varepsilon$ is from the arrivals. 
 We test the impact of   traffic models       in   Figure~\ref{fig:delaybound}.      The first line of Figure~\ref{fig:delaybound} displays the result with $c=100Mbps$, and the second line, $c=35Mbps$.   We   see the same phenomenon   here: the reliability is not guaranteed when we choose an inappropriate model to do   feature extraction.   Take the cPoisson model  as  an instance. It can provide accurate results in high-bandwidth scenarios, whereas it behaves dramatically poor as bandwidth decreases.   MGF-SNC with
 the other arrival models might severely  underestimate the delay bound.  By contrast, 
 \sys   can  accurately  capture the various characteristics of traffic  from the decomposed spatio-temporal signals and hence significantly boost the effectiveness of MGF-SNC, which means that both the tightness and the reliability of the delay bound are enhanced (see also Figure~\ref{fig:bandwidth} for the stochastic delay bound with violation probability $\varepsilon=10^{-4}$ and $\varepsilon=10^{-5}$, respectively).


Different  from  flow 1, flow 2 and 3's end-to-end services are influenced by the cross-flows.  When the  flow  of interest  merges with a cross-flow, its service might be strongly reduced.  Compared to its deterministic counterpart,   
the impact of cross-traffic   in   SNC  can be stronger,  as computation   operations like leftover service or deconvolution require several stochastic inequalities, which   in many cases leads to loose bounds.  Since our contribution is not on the service modeling,  we will not pay a lot of discussion on how to reduce pessimism brought by SNC in complex topologies and leave
this to future work.  What we want to address here is that even if more advanced modeling of cross-flow is employed,    MGF-SNC with existing arrival models  might  still underestimate the delay bound (see the last two columns of  Figure~\ref{fig:delaybound}).   The proposed \sys model outperforms popular arrival models   and shows robustness under different scenarios.  The main credit for that goes to the versatility of \sys which contributes to better extracting critical traffic features from  network flows. However, the other arrival models may not work well. As network utilization level grows, complicated characteristics of network traffic  gradually play a more influential role. Existing models are too short-sighted to capture such characteristics  of real traffic and hence cannot provide an appropriate estimation for performance bound. The proposed \sys accurately captures such characteristics, and hence improves performance analysis.

%% file: sections/5.conclusion.tex
\section{Conclusion}\label{sec:conclucion}
  In this paper, we revisit the traffic modeling problem in SNC and propose a spatial-temporal model \sys for real-world network traffic. We dissect various existing arrival models and reveal that \sys unifies these models. Then we derive the MGF-bound of \sys which can be directly integrated into the framework of MGF-SNC. Extensive experiments with real traffic traces demonstrate that \sys can accurately represent multi-dimensional multi-order characteristics of current network traffic. Experiments also show that MGF-SNC with \sys achieves a tighter and more robust performance bound, while the existing arrival models are either  overly optimistic  for accurate performance evaluation or sensitive to specific traffic types and scenarios. In a nutshell, \sys provides a fine-grained QoS guarantee and facilitates effective network planning. We envision that   \sys   will motivate further exploration of traffic modeling for improving SNC.

%% file: sections/6.appendix.tex
\section*{Appendix}
\subsection*{Proof of Theorem~\ref{cmgf}}
The proof uses a backwards mathematical induction that starts with $s=t-1$. Let $\mathbb{E}_s[\cdot]=\mathbb{E}[\cdot|Z_s]$.  When $s=t-1$, \\ 
$$
\small{ 
\begin{array}{ll}
 \quad\mathbb{E}_{t-1}[e^{\theta A(t-1,t)}]=\mathbb{E}_{t-1}[e^{\theta a_t}]\\
	=\mathbb{E}_{t-1}[e^{\theta[\breve\mu(Z_{t-1}, Z_t)+\sum_{i=1}^p\breve\phi_i(Z_{t-1},Z_t)y_{N(t-1)+1-i}+\breve\sigma(Z_{t-1}, Z_t)\epsilon_t]}]\\
	=\mathbb{E}_{t-1}[e^{\theta\sum_{i=1}^p\breve\phi_i(Z_{t-1}, Z_t)y_{N(t-1)+1                                                                                                                                                                                                                                                                                                           -i}}M(\theta;  Z_{t-1}, Z_t)].
\end{array}
}
$$ 
Since $\theta_1(t,t)=\theta$  and 
$
\theta_i(t-1, t)-\theta=\theta[\varphi_i(t-1,t)-1]=\theta\breve\phi_i(Z_{t-1}, Z_t),
$
Theorem~\ref{cmgf} holds for the case of $s=t-1$.

Assume that Theorem~\ref{cmgf} holds for $A(s,t)$. Then we have 
$$
\begin{array}{ll}
	\quad \mathbb{E}_{s-1}[e^{\theta A(s-1, t)}]\\
	= \mathbb{E}_{s-1}[e^{\theta a_s+\theta A(s, t)}]\\
	=  \mathbb{E}_{s-1}[e^{\theta a_s}\mathbb{E}_s[e^{\theta A(s,t)}]]\\
	=    \mathbb{E}_{s-1}[e^{\theta a_s+\sum_{i=1}^p[\theta_i(s,t)-\theta]y_{N(s)+1-i}}\\
	  \;\;\;\;\;\;\;\;\times\prod_{l=s+1}^tM(\theta_1(l,t); Z_{l-1}, Z_l)]\\
	= \mathbb{E}_{s-1}[e^{[\theta_1(s,t)-\theta(1-O_s)]y_{N(s)}+\sum_{i=2}^p[\theta_i(s,t)-\theta]y_{N(s)+1-i}}\\
	 \;\;\;\;\;\;\;\;\times\prod_{l=s+1}^tM(\theta_1(l,t); Z_{l-1}, Z_l)].
\end{array}
$$

   Noticing that:     (1) $
   		y_{N(s)}=(1-O_s)y_{N(s-1)}+\breve\mu(Z_{s-1}, Z_s)+\sum_{i=1}^p\breve\phi_i(Z_{s-1}, Z_s)y_{N(s-1)+1-i}+\breve\sigma(Z_{s-1}, Z_s)\epsilon_s$; and (2)  
     $y_{N(s)+1-i}=O_sy_{N(s-1)+2-i}+(1-O_s)y_{N(s-1)+1-i}$.
Direct  substitution   yields 
$$
\small{
\begin{array}{ll}
\quad\mathbb{E}_{s-1}[e^{\theta A(s-1,t)}]\\
=\mathbb{E}_{s-1}\big[e^{\sum_{i=1}^p[\theta_i(s-1,t)-\theta]y_{N(s-1)+1-i}}  \prod_{l=s}^t M(\theta_1(l,t); Z_{l-1}, Z_l)\big]
\end{array}
}
$$
which completes the proof.

$$\quad$$
\subsection*{Proof of Theorem~\ref{clm}} 
 Let $\{\mathcal{F}_t\}_t$ be   the complete filtration generated
 by $Z$. We first have      
 $ 
 W_s=b(Z_{s+1})+A(Z_s, Z_{s+1})W_{s+1},
$
where $W_s$ is a random variable adapted to $\mathcal{F}_{s-1}$.
 Taking the expectation with $\mathcal{F}_s$ on both sides yields the following equation for $\bm{v}$: 
  $
  \bm{v}=T B+ T \circ A\bm{v}. 
  $

As to the second claim, we first have       
   \begin{claim}\label{clmproof} Let $||\cdot||_\infty$ be the infinity norm, and 
   	$$
   	\breve{\phi}(i,j)=\begin{bmatrix}
   		\breve{\phi}_1(i,j)\\
   		\vdots\\
   		\breve{\phi}_p(i,j)
   	\end{bmatrix}.
   $$ 
   We then have 
   	$|\varphi(s,t)-v_{Z_s}| \le c_{s,t} \varepsilon$, where 
   	\begin{itemize}
   		\item $c_{t-1,t}=\bm{1}$;
   		\item $c_{s,t} =\bm{1}+ abs(A(Z_s, Z_{s+1}))    c_{s+1,t}, \; \forall s<t-1$;  and   
   		\item $\varepsilon=\max\Big(
   		\begin{array}{l}
   			\max_{i,j}||\bm{1}+ \breve{\phi}(i,j)  -v_i||_\infty, \\
   			\max_{i,j}||b(j)+A(i,j)v_j-v_i||_\infty
   		\end{array}
   		\Big)$.
   	\end{itemize}
   \end{claim}
 	\textit{Proof of  Claim~\ref{clmproof}.} The proof uses a backwards mathematical induction that starts with $s=t-1$. When $s=t-1$, 
 	$$
  \varphi(t-1, t)  
 	 = \bm{1}+\breve{\phi}(Z_{t-1}, Z_t) 
 	$$ 
 	By definition, we have $  |\varphi(t-1, t)-v_{Z_{t-1}}| \le \varepsilon \bm{1}$.  
 	
 	Assume that claim~\ref{clmproof} holds for $\varphi(s+1, t)$.  
 	We then have
 	$$
 	\begin{array}{lll}
 		&& |\varphi(s, t)-v_{Z_s}| \\
 		&=& |b(Z_{s+1})+A(Z_s, Z_{s+1})\varphi(s+1, t)-v_{Z_s}|  \\
 		&\le &|b(Z_{s+1})+A(Z_s, Z_{s+1})v_{Z_{s+1}}-v_{Z_s}|\\
 		&&+abs(A(Z_s, Z_{s+1}))  c_{s+1,t}\varepsilon\\
 		&\le &[\bm{1}+abs(A(Z_s, Z_{s+1}))c_{s+1, t}]\varepsilon\\
 		&\le&c_{s,t} \varepsilon, 
 	\end{array}
 	$$
 	which completes the proof.  \\

Let $\epsilon_t\sim \mathcal{N}(0,1)$\footnote{The case of $\epsilon_t\sim{\rm Exp}(1)$ can be similarly  argued.  },  according to the result of Claim~\ref{clmproof}, we   have 
\begin{enumerate}
	\item 
$$
\begin{array}{ll}
	\quad e^{\sum_{i=1}^p[\theta_i(s,t)-\theta]y_{N(s)+1-i}}\\
 \le e^{\sum_{i=1}^p\theta[v_{Z_s,i}-1+c^{max}_{s,t}\varepsilon ]y_{N(s)+1-i}},
\end{array}
$$ 
where $c^{max}_{s,t}=\sup||c_{s,t}||_\infty$;   
\item 
$$
 \begin{array}{lll}
 	\quad  \mathbb{E}[\prod_{l=s+1}^tM(\theta_1(l,t); Z_{l-1}, Z_{l})|Z_s=k]\\
 = \mathbb{E}[\prod_{l=s+1}^t e^{a_l\varphi_1(l, t)+b_l\varphi_1(l, t)^2}|Z_s=k]\\
    \le   \mathbb{E}[\prod_{l=s+1}^t e^{a_l(v_{Z_l,1}+c_{l,t,1}\varepsilon)+b_l(v_{Z_l,1}+c_{l,t,1}\varepsilon)^2}|Z_s=k]\\
    \le  e^{(h_1\varepsilon+h_2\varepsilon^2)(t-s-1)}\times\\
     \quad \underbrace{\mathbb{E}[\prod_{l=s+1}^te^{a_lv_{Z_l,1}+b_lv_{Z_l,1}^2}|Z_s=k]}_{\Lambda_{s, t}},
 \end{array}
 $$
 where 
 $$
 \left\{  
 \begin{array}{lll}
 	a_i= \breve\mu(Z_{i-1}, Z_i)\theta\\
 	b_i=\frac{1}{2}\breve\sigma(Z_{i-1}, Z_i)^2\theta^2\\
 	h_1=\max_i  (a_i+2b_iv_{Z_i,1})c_{i,t,1}\\
 	h_2=\max_i  b_ic_{i,t,1}^2
 \end{array}
 \right.
 $$
  \end{enumerate} 
Consider the building block $\Lambda_{s, t}$. 
 Let $\bm{J}(s,t)$ be the vector $[\mathbb{E}[\Lambda_{s,t}|Z_s=k]]_{k\in\mathcal{S}}$, 
 we then have
  $$
  \begin{array}{lll}
  	J_k(s,t)&=&\mathbb{E}[\Lambda_{s,t}|Z_s=k]\\
  	&=&\mathbb{E}[e^{a_{s+1}v_{Z_{s+1},1}  +b_{s+1}v_{Z_{s+1},1}^2}\Lambda_{s+1, t}|Z_s=k]\\
  &=&(\Gamma(\theta)J(s+1,t))_k.
\end{array}
$$
And then,    
$$
\bm{J}(s,t)=\Gamma(\theta)J(s+1, t)=\cdots=\Gamma^{t-s}(\theta)\bm{e}.
$$
Therefore,     
 $$
 \begin{array}{lll}
&& \mathbb{E}[e^{\theta A(s,t)}]\\
&\le&\bm{\pi}[\mathbb{E}[e^{\theta A(s,t)}|Z_s=k]]_{k\in\mathcal{S}}\\  &\le&e^{[h_1\varepsilon+h_2\varepsilon^2](t-s-1)+\sum_{i=1}^p\theta c^{max}_{s,t}  y_{N(s)+1-i} \varepsilon}\bm{\pi}{\rm R}J(s,t).
 \end{array}
 $$
 
Note that   when  $\phi_{\max}=0$,     $v_i\equiv \bm{e},   \forall i\in\mathcal{S}$, we then have     $\varepsilon=0$.   Since $\varepsilon$ is continuous at   $\phi_{\max}=0$, we  then have
 $$
 \begin{array}{ll}
 	   &\lim_{\phi_{\max}\to0} c_{s,t}\varepsilon\\
 	   =& \lim_{\phi_{\max}\to0}c_{s,t}\times \lim_{\phi_{\max}\to0}\varepsilon\\
 	=& (t-s)\times \lim_{\phi_{\max}\to0}\varepsilon\\
 	=&0,
 \end{array}
 $$
which completes the  proof.

%% file: sections/7.ethics.tex
\subsection*{Ethics}
 This work does not raise any ethical issues.